\documentclass[a4paper,11pt]{article}
\usepackage{cite}
\usepackage[numbers,sort&compress]{natbib}
\usepackage{txfonts}
\usepackage{bbding}
\usepackage{pifont}
\usepackage{mathbbol}

\usepackage[fleqn]{amsmath}
\newtheorem{myremark}{\bf  Remark}

\newtheorem{mydefinition}{\bf Definition}

\newtheorem{mytheorem}{\bf Theorem}
\newtheorem{mylemma}{\bf Lemma}

\usepackage{enumerate}
\usepackage{amsfonts}
\usepackage{tipa}
\usepackage{amsmath}
\usepackage{amssymb}
\usepackage{graphicx}
\usepackage{epsfig}
\usepackage{subfigure}
\usepackage{subfig}
\usepackage{indentfirst}
\usepackage{setspace}
\usepackage{bm}
\usepackage{multicol,float}
\makeatletter
\newcommand{\rmnum}[1]{\romannumeral #1}
\newcommand{\Rmnum}[1]{\expandafter\@slowromancap\romannumeral #1@}
\newenvironment{proof}[1][Proof]{\noindent\textbf{#1.} }{\ \rule{0.5em}{0.5em}}
\makeatother

\begin{document}
\baselineskip 0.7 cm
\title{\normalsize\bf Advances in Difference Equations, 2018；\\
DOI: 10.1186/s13662-019-2074-8\\
Modeling, discretization, and hyperchaos detection of conformable derivative approach to a financial system with market confidence and ethics risk}
\author{\qquad Baogui Xin\thanks{B. Xin, Nonlinear Science Center, College of Economics and Management, Shandong University of Science and Technology, Qingdao 266590, China, e-mail: xin@tju.edu.cn, corresponding author} \qquad\ Wei Peng\thanks{%
W. Peng, Nonlinear Science Center, College of Economics and Management, Shandong University of Science and Technology, Qingdao 266590, China, e-mail: pengweisd@foxmail.com} \qquad\ Yekyung Kwon\thanks{%
Y. Kwon, Division of Global Business Administration, Dongseo University, Busan 47011, Korea, e-mail: yiqing@hanmail.net} \qquad\  Yanqin Liu%
\thanks{%
Y. Liu, College of Mathematical Sciences, Dezhou University, Dezhou 253023, China, e-mail: yqliumath@163.com} }
\date{ }
\maketitle

\begin{abstract}
\baselineskip 0.7 cm A new chaotic financial system is proposed by considering ethics involvement in a four-dimensional financial system with market confidence. A five-dimensional conformable derivative financial system is presented by introducing conformable fractional calculus to the integer-order system. A discretization scheme is proposed to calculate numerical solutions of conformable derivative systems. The scheme is illustrated by testing hyperchaos for the system.


\end{abstract}

\section{Introduction}
A hyperchaotic system is typically defined as a chaotic system with at least two positive Lyapunov exponents\cite{LiuWei2018,WeiMoroz2018,WeiRajagopal2018}. Considering one null exponent along the flow and one negative exponent to ensure the boundedness of the solution, the dimension for a hyperchaotic fractional-order system is at least four. Many types of fractional-order hyperchaotic systems have been investigated, such as hyperchaos in fractional-order R$\ddot{o}$ssler systems \cite{lichen01}, fractional hyperchaotic systems with flux controlled memristors \cite{Rajagopal01}, fractional-order hyperchaotic systems realized in circuits \cite{El-Sayed01,El-Sayed04,Mouj01}, fractional-order Lorenz hyperchaotic systems \cite{WangY01}, and fractional-order hyperchaotic cellular neural networks \cite{HuangX01}. Huang $\&$ Li \cite{Huang} proposed an interesting nonlinear financial system depicting the relationship among interest rates, investments, prices, and savings. Chen \cite{Chenw1} presented a fractional form of nonlinear financial system. Wang, Huang and Shen \cite{Wang3} established an uncertain fractional-order form of the financial system. Mircea et al. \cite{Mircea1} set up a delayed form of the financial system. Xin, Chen, and Ma \cite{Xin1} proposed a discrete form of financial system.  Yu, Cai, and Li \cite{yuu} extended the financial system with the average profit margin. Xin, Li, and Zhang \cite{Xin07,XinB01} introduced investment incentive and market confidence to the nonlinear financial system to set up novel four-dimensional financial systems. Most of these are fractional-order hyperchaotic systems \cite{ZhangL01,WangS01}. In this paper, we will construct a 5D fractional-order hyperchaotic financial system.

Although Riemann-Liouville, Caputo, and Grunwald-Letnikov fractional derivatives \cite{JiangX01,JiangX02,QiH01,QiH02,ChenD01,ChenD02,WuG01,WuG02} are widely used in physics, mathematics, medicine, economics, and engineering as shown above, these derivative definitions lack some of the agreed properties for  classical differential operator, such as the chain rule. The conformable derivative can be regarded as a natural extension of the classical differential operator, which satisfies most important properties, such as the chain rule \cite{KhalilR01,Abdeljawad01,AbdeljawadA02}. Researchers have recently applied conformable derivatives to many scientific fields \cite{AcanF01,AttiaL01,BohnerH01,Tarasov01,RosalesG01,AkbulutM01,MartR01,RezazadehK01,Korkmaz01,PerezG01,HeBY01}. In this paper, we will introduce the conformable derivative to the financial system with market confidence and ethics risk.

Researchers \cite{LuYang01,ErfaniV01,DincerY01,LiLiu01,CavdarA01,ZhaoZhao01} have paid much attention to the impact of confidence on the financial system since the financial crisis began in 2008. Financial crises are sometimes associated with a vicious circle in which confidences and economic indicators interact. For example, lower consumer confidence can lead to weaker consumption expectation, which will bring about a sharp decrease in market demand, stagnation in orders, and a marked decrease in sales. This can hold down investor confidence and trigger a decrease in employment and wages, which also lead to lower consumer confidence. Alternatively, we can boost the demands of consumers and investors by improving individual incomes, confidences, and expectations. Certainly, investor confidence can be rescued by lowering interest rates. Consumers' balance sheets can result in a price index, which will affect the confidences of both investors and consumers \cite{XinB01}. Thus, it is necessary to study financial systems with market confidence.

As for mainstream finance theory, financial participants' profit drive and ethics-denunciation can rub together, so the profit motive leads them to crimp on ethics and disregard the broader social impact of their actions. Unrestrained demand for profit turns financial participants into psychopaths. Although the profit-driven segment can earn abnormal returns through the violation of ethics in the short run, these profit-generating opportunities do not persist in the long run for the socially responsible investment movement \cite{Derwall01}, because it violates the inherent requirement of financial efficiency, entailing a high risk that the financial market falls into instability. Analogous to Rasmussen's logic \cite{Rasmussen01}, the drive for profit is indeed a double-edged sword. As has long been recognized, it is an inevitable result of flourishing financial markets as a positively useful means of encouraging investment. Conversely, it has some of the other effects of extreme profit-driven behaviors, which leads people to sympathize more fully and readily with selfish profit maximization, as opposed to altruistic social responsibility, and this distortion in our sympathies, in turn, undermines both investment ethics and public welfare. Ethics may be an easy for many subject to understand, but to implement it in financial markets requires faith, dedication, determination, and regulation. Profit-driven behavior in a market economy is not naturally reasonable, but it must be restrained. Therefore, it is meaningful for us to consider ethics when we analyze the financial system.

After we propose a conformable derivative financial system, we need a suitable scheme to obtain its solutions. Though there are several methods to solve a conformable derivative system \cite{EslamiR01,IlieB01,HosseiniB01,UnalG01,KumarS01,SrivastavaG01,KaplanA01,YavuzM01,KartalS01,IyiolaT01,PerezG01,RuanS01,HeSun01,Yokus01,RezazadehZ01,ZhongW01,TayyanS01,YaslanN01,KurtC01,KhalilA01,UnalG02,LiuW01,CenesizK01}, but these are too complex for many people. Inspired by the discretization process for Caputo derivative \cite{El-Sayed02,Agarwal01}, we propose a simple discretization process for a conformable derivative. As shown in Section 3, our results agree with Mohammadnezhad's conformable Euler's method \cite{Mohammadnezhad01}. Using our proposed discretization scheme, we will detect the hyperchaotic attractor of a five-dimensional fractional-order financial system.

The remainder of this paper is organized as follows. Section 2 presents a conformable derivative hyperchaotic financial system with market confidence and ethics risk. Section 3 provides a conceptual overview of conformable calculus and propose a conformable discretezation process, which coincides with Mohammadnezhad's conformable Euler's method \cite{Mohammadnezhad01}. In section 4, we detect the hyperchaotic attractor from the proposed financial system. Some concluding remarks in Section 5 close the paper.

\section{A conformable derivative hyperchaotic financial system with market confidence and ethics risk}
\begin{figure}
\centerline { \epsfig{figure=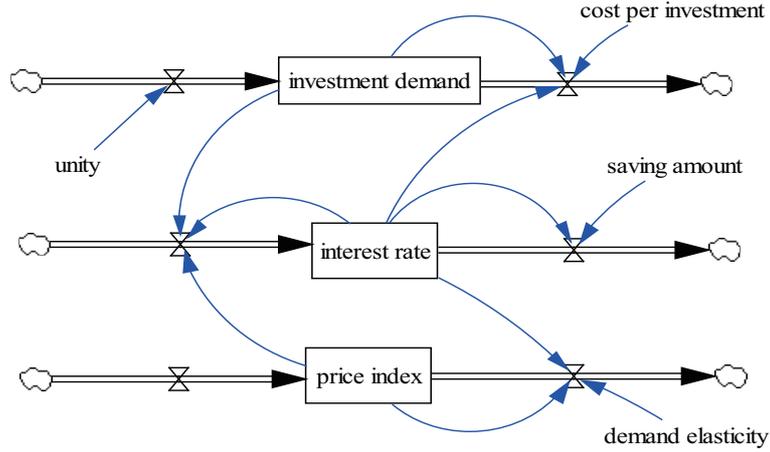, height=60mm,width=100mm}}
\caption{{\small Schematic diagram of a simple financial system (\ref{eq01}). }} \label{fg001}
\end{figure}
Based on the dynamical mechanism of financial systems, as shown in Figure \ref{fg001},  Huang $\&$ Li \cite{Huang} proposed an interesting nonlinear financial system, as follows:
\begin{equation} \label{eq01}
\left\{ \begin{aligned}
\dot{x}&= z+(y-a)x,\\
\dot{y}&= 1-by-x^{2},\\
\dot{z}&= -x-cz,\\
\end{aligned} \right.
\end{equation}
where $x, y, z, a, b, c$ are the interest rate, investment
demand, price index, saving amount, cost per investment, and demand elasticity of commercial markets, respectively, and $a,\ b,\ c\geq 0$.

\begin{figure}
\centerline { \epsfig{figure=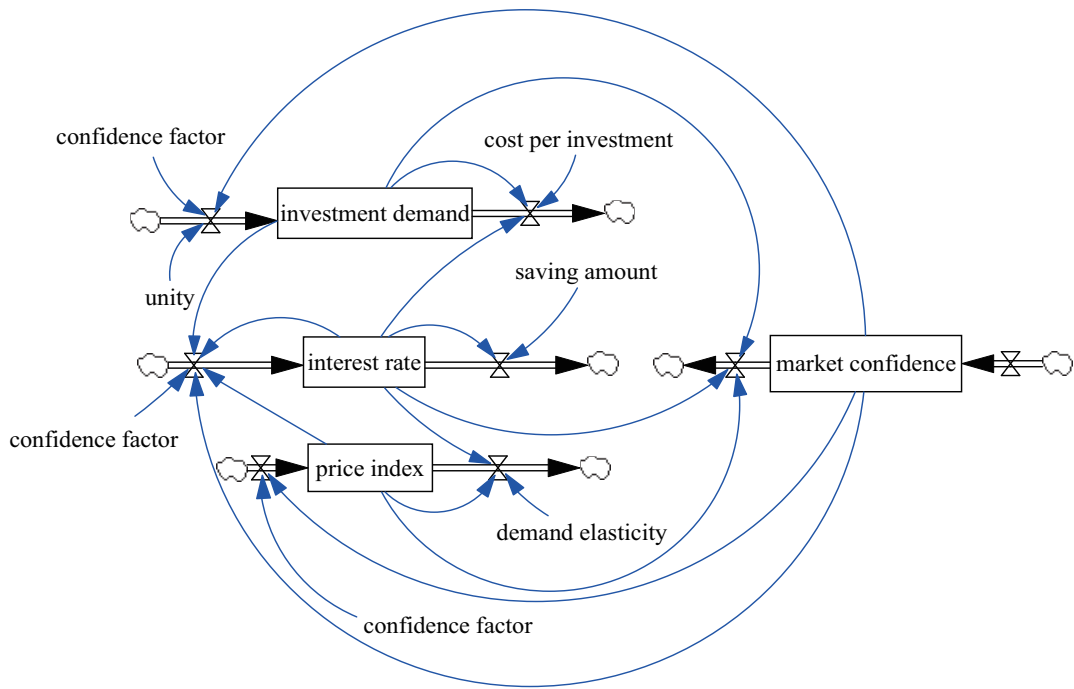, height=90mm,width=140mm}}
\caption{{\small Schematic diagram of financial system (\ref{eq02}) with market confidence. }} \label{fg002}
\end{figure}

Xin and Zhang \cite{XinB01} took into account the market confidence in system (\ref{eq01}), updated the dynamical mechanism of financial systems with the market confidence, as shown in Figure \ref{fg002}, and obtained the following financial system:
\begin{equation} \label{eq02}
\left\{ \begin{aligned}
\dot{x}&= z+(y-a)x+m_{1}w,\\
\dot{y}&= 1-by-x^{2}+m_{2}w,\\
\dot{z}&= -x-cz+m_{3}w,\\
\dot{w}&=-xyz,\\
\end{aligned} \right.
\end{equation}
where $x, y, z, a, b, c $ are defined as in system (\ref{eq01}), $ w $ indicates the market confidence, and $m_{1},\ m_{2},\ m_{3} $ are the impact factors.

Ethics risk is often portrayed as a threat under asymmetric information. Uncertain or incomplete contracts cause the responsible stakeholders to not bear all of the consequences for maximizing their own utility while harming other stakeholders through improper behavior such as lying, cheating, and breaking the terms of a contract. However, there are no clear guidelines for governments to deal with ethical dilemmas, whether in law or often in religion. Thus, ethical risk may occur when stakeholders must choose among alternatives, for example, when significant value conflicts exit among differing interests, real alternatives with justifiable equality and mutual benefit, significant consequences on them. Thus we can update the dynamical mechanism of financial system (\ref{eq02}) by accounting for ethical risk, as shown in Figure \ref{fg003}.

\begin{figure}
\centerline { \epsfig{figure=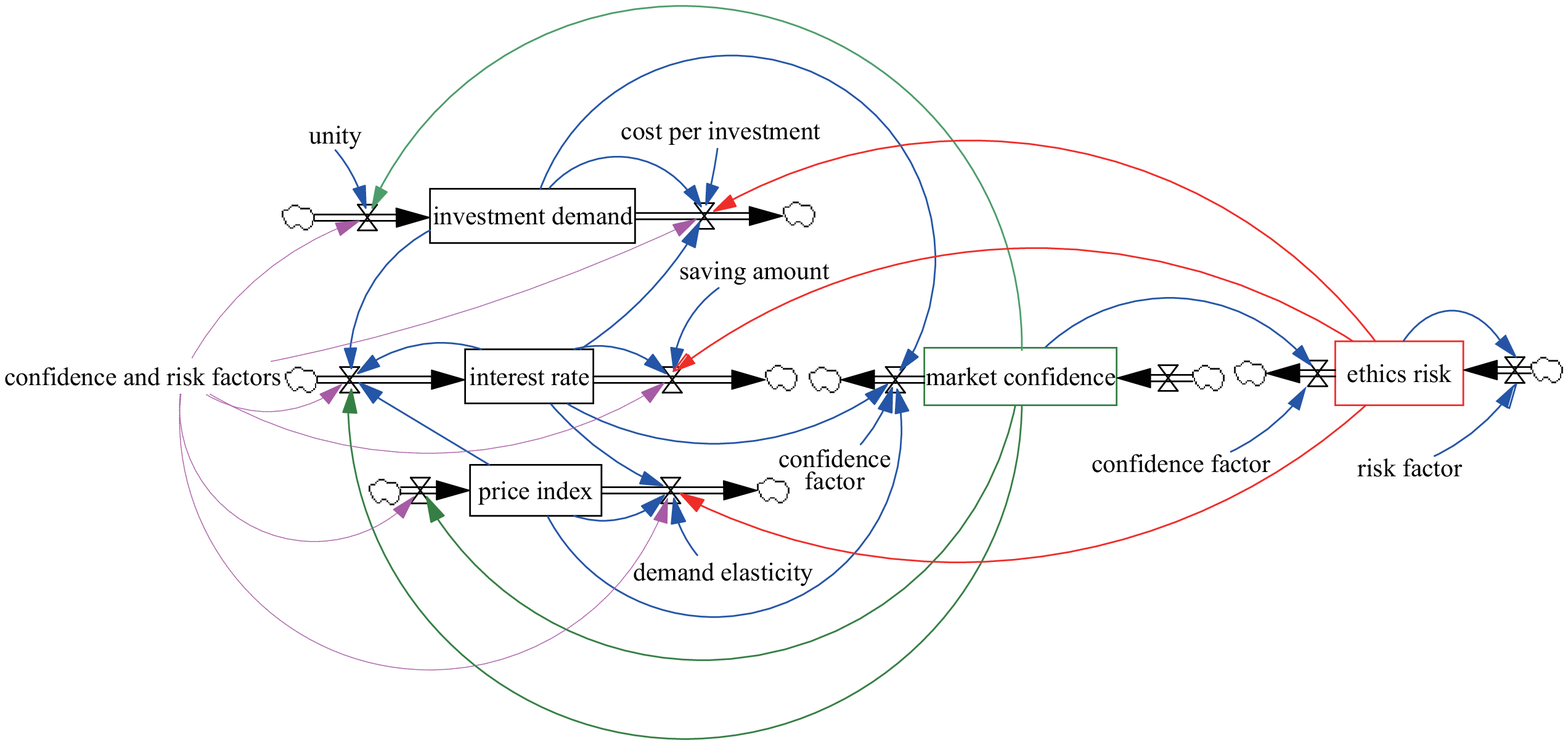, height=90mm,width=160mm}}
\caption{{\small A schematic diagram of the financial system (\ref{eq03}) with market confidence and ethics risk. }} \label{fg003}
\end{figure}

Ethics risk can negatively affects market confidence. They will to some extent offset the impact of market confidence on the interest rate, investment demand, and price index. In addition, the higher the market confidence, the less the motivation to harm other stakeholders, and the lower the ethics risk. What is more, ethical risks becoming increasingly lower with the continuous improvement of regulations. Market confidence is influenced by many factors, including ethics, law, and religion, so we will consider the impact of ethics risk on market confidence. Thus we can obtain the following financial system accounting for both market confidence and ethics risk:
\begin{equation} \label{eq03}
\left\{ \begin{aligned}
\dot{x}&= z+(y-a)x+k(w-pu),\\
\dot{y}&= 1-by-x^{2}+k(w-pu),\\
\dot{z}&= -x-cz+k(w-pu),\\
\dot{w}&=-dxyz,\\
\dot{u}&=k(w-pu),\\
\end{aligned} \right.
\end{equation}
where $x, y, z, w, a, b, c $ means the same as in the system (\ref{eq02}), and $ u $ denotes the ethics risk, and $k,\ p,\ d $ are the impact factors.

Fractional-order economic systems \cite{Xin2,Yavuz01,Baskonus01,MaRen01,HuangW01,XinB01} can generalize its integer-order forms \cite{Xin07,Xin4}. As a natural extension of the integer-order differential operator, the conformable fractional operator is a suitable tool to generalize integer-order forms, so we will introduce conformable fractional derivatives to financial system (\ref{eq03}), as follows:

\begin{equation} \label{eq04}
\left\{ \begin{aligned}
T_{\alpha_{1}}x&= z+(y-a)x+k(w-pu),\\
T_{\alpha_{2}}y&= 1-by-x^{2}+k(w-pu),\\
T_{\alpha_{3}}z&= -x-cz+k(w-pu),\\
T_{\alpha_{4}}w&= -dxyz,\\
T_{\alpha_{5}}u&= k(w-pu),\\
\end{aligned} \right.
\end{equation}
where $\alpha=(\alpha_{1},\alpha_{2},\alpha_{3},\alpha_{4},\alpha_{5})$ is subject to $
\alpha_{1},\alpha_{2},\alpha_{3},\alpha_{4},\alpha_{5}\in(0,1)$.

\begin{myremark}
When $\alpha=(1,1,1,1,1)$, system (\ref{eq04}) degenerates to system (\ref{eq03}).
\end{myremark}

\section{Discretization of conformable derivative systems }
\subsection{Preliminary}

\begin {mydefinition}(See \cite{Abdeljawad01})\label{r01}
 For a function $f:[t_{0}, \infty)\rightarrow \mathbb{R}$, its left conformable derivative starting from $t_{0}$ of order $\alpha\in(0,1)$ is defined by
\begin{equation}\label{fe01}
T_{\alpha}^{t_{0}}f(t)=\lim_{\varepsilon\rightarrow 0}\frac{f(t+\varepsilon (t-t_{0})^{1-\alpha})-f(t)}{\varepsilon},
\end{equation}
\end {mydefinition}
where the function $f$  is called $\alpha$-differentiable.

\begin {mydefinition}(See \cite{Abdeljawad01})\label{r02}
For a function $f:[t_{0}, \infty)\rightarrow \mathbb{R}$, its left conformable integral starting from $t_{0}$ of order $\alpha\in(0,1)$ is defined as
\begin{equation}\label{fe02}
I_{\alpha}^{t_0}f(t)=\int^{t}_{t_0}(s-t_0)^{\alpha-1}f(s)ds,
\end{equation}
where the integral is the usual Riemann improper integral.
\end {mydefinition}

\begin {mylemma}(See \cite{Abdeljawad01}) \label{r03}
Suppose the derivative order $\alpha\in(0,1)$, and functions $f$ are $\alpha$-differentiable at a point $t_0>0$. Then the left conformable derivative satisfies
\begin{equation}\label{fe03}
T_{\alpha}^{t_{0}}f(t)=(t-t_0)^{1-\alpha}\frac{df(t)}{dt}.
\end{equation}
\end {mylemma}

\begin {mylemma}(Conformable Euler's method)\cite{Mohammadnezhad01} \label{r05}
Consider the following conformable derivative system:
\begin{equation}\label{fe04}
T_{\alpha}x(t)=f(x(t)),\quad 0\leq t\leq T,\quad x(0)=x_0,
\end{equation}
where $h=\frac{T}{N}=t_{n+1}-t_{n}$, $t_n = nh,\ n=0,1,\cdots,N$.
If $h$ is small enough and $T_{\alpha}x(t),\ T_{2\alpha}x(t)\in C^{0}[a,b]$, then we can obtain the following discretization of Eq.(\ref{fe04})
\begin{equation}\label{fe05}
x_{n+1}\approx x_n + \frac{h^{\alpha}}{\alpha}f(x_n).
\end{equation}
\end {mylemma}

\subsection{Discretization process}
By introducing partially piecewise constant arguments to Eq. (\ref{fe04}), Kartal and Gurcan\cite{KartalS01} obtain the form
\begin{equation}\label{fe010}
T_{\alpha}x(t)=f(x(t),x(\left[\frac{t}{h}\right]h)),\quad 0\leq t\leq T,\quad x(0)=x_0,
\end{equation}
where $h=\frac{T}{N}$, i.e., $t\in[nh,(n+1)h),\ n=0,1,2,\cdots \frac{T}{h}$.
Then Kartal and Gurcan \cite{KartalS01} proposed the discretization process of system (\ref{fe010}), but it was unsuitable for the following form obtained by introducing completetely piecewise constant arguments to Eq. (\ref{fe04}).
\begin{equation}\label{fe10}
T_{\alpha}x(t)=f(x(\left[\frac{t}{h}\right]h)),\quad 0\leq t\leq T,\quad x(0)=x_0.
\end{equation}

\begin{mytheorem}\label{r06}
(Conformable discretization by piecewise constant approximation)

According to system (\ref{fe10}), we obtain the following discretization of Eq.(\ref{fe04})
\begin{equation}\label{fe11}
x_{n+1}= x_n + \frac{h^{\alpha}}{\alpha}f(x_n),
\end{equation}
where $x_n$ denotes $x_n(nh)$.
\end{mytheorem}

\begin{proof}
Using Lemma \ref{r03}, we rewrite Eq. (\ref{fe10}) as
\begin{equation*}
(t-nh)^{1-\alpha}\frac{dx(t)}{dt}=f(x(nh)),\quad 0\leq t\leq T,\quad x(0)=x_0,
\end{equation*}
which leads to
\begin{equation}\label{fe12}
\frac{dx(t)}{dt}=(t-nh)^{\alpha -1}f(x(nh)),\quad 0\leq t\leq T,\quad x(0)=x_0.
\end{equation}

Drawing on the step method \cite{El-Sayed02,KartalS01}, we detail the steps of the discretization process:

(\rmnum{1}) Let $n=0$, then $t\in[0,h)$, and we rewrite Eq. (\ref{fe12})
\begin{equation}\label{fe13}
\frac{dx(t)}{dt}=(t-0)^{\alpha -1}f(x_0),\quad t\in[0,h),
\end{equation}
and the solution of Eq. (\ref{fe13}) is
\begin{equation*}
\begin{aligned}
x_1(t)&=x_0 +\int_0^{t}\left((s-0)^{\alpha -1}f(x_0)\right)ds\\
&=x_0 +f(x_0)\int_0^t s^{\alpha-1}ds\\
&=x_0 +\frac{t_{\alpha}}{\alpha}f(x_0)).
\end{aligned}
\end{equation*}
(\rmnum{2}) Let $n=1$. Then $t\in[h,2h)$, and we rewrite Eq. (\ref{fe12})
\begin{equation}\label{fe14}
\frac{dx(t)}{dt}=(t-h)^{\alpha -1}f(x_1(h)),\quad t\in[h,2h),
\end{equation}
and the solution of Eq. (\ref{fe14}) is
\begin{equation*}
\begin{aligned}
x_2(t)&=x_1(h) +\int_h^{t}\left((s-h)^{\alpha -1}f(x_1(h))\right)ds\\
&=x_1(h) +f(x_1(h))\int_h^t (s-h)^{\alpha-1}ds\\
&=x_1(h) +\frac{(t-h)^{\alpha}}{\alpha}f(x_1(h)).
\end{aligned}
\end{equation*}
(\rmnum{3}) By repeating the above process, we obtain the following solution of Eq. (\ref{fe12}):
\begin{equation*}
x_{n+1}(t)=x_{n}(nh)+\frac{(t-nh)^{\alpha}}{\alpha}f(x_n(nh)),\quad t\in[nh,(n+1)h).
\end{equation*}
Let $t\rightarrow (n+1)h$. We deduce the the following discretization:
\begin{equation*}
x_{n+1}((n+1)h)=x_{n}(nh)+\frac{h^{\alpha}}{\alpha}f(x_n(nh)),\quad t\in[nh,(n+1)h).
\end{equation*}
That is,
\begin{equation*}
x_{n+1}=x_{n}+\frac{h^{\alpha}}{\alpha}f(x_n).
\end{equation*}
It is proved.
\end{proof}
\begin{myremark}
The conformable discretization by piecewise constant approximation well coincides with the conformable Euler's method \cite{Mohammadnezhad01}.
\end{myremark}

\section{Hyperchaos detection}
Using Theorem \ref{r06}, we rewrite system (\ref{eq04}) using piecewise constant approximation, as follows:
\begin{equation} \label{eq4444}
\left\{ \begin{aligned}
x_{n+1}&= x_n + \frac{h^{\alpha_{1}}}{\alpha_{1}} (z+(y-a)x+k(w-pu)),\\
y_{n+1}&= y_n + \frac{h^{\alpha_{2}}}{\alpha_{2}} (1-by-x^{2}+k(w-pu)),\\
z_{n+1}&= z_n + \frac{h^{\alpha_{3}}}{\alpha_{3}} (-x-cz+k(w-pu)),\\
w_{n+1}&= w_n + \frac{h^{\alpha_{4}}}{\alpha_{4}} (-dxyz),\\
u_{n+1}&= u_n + \frac{h^{\alpha_{5}}}{\alpha_{5}} k(w-pu).\\
\end{aligned} \right.
\end{equation}

In this section, we implement hyperchaos detection by varying the parameters related to ethics risk, such as $\alpha_5$, the confidence factor $k$, and the risk factor $p$. To detect hyperchaos in system (\ref{eq4444}) using conformable discretization process, we fix the following parameters and initial point values: $h=0.002,\ a=0.8,\ b=0.6,\ c=1,\ d=2,\ \alpha_1=0.3,\ \alpha_2=0.5,\ \alpha_3=0.6,\ \alpha_4=0.24,$ $x_0=0.4,\ y_0=0.6,\ z_0=0.8,\ w_0=0.3,\ u_0=0.4.$

\begin{figure}
\begin{center}
\begin{minipage}{140mm}
\subfigure[Lyapunov exponents vs. $\alpha_5$.]{
\resizebox{7cm}{!}{\includegraphics{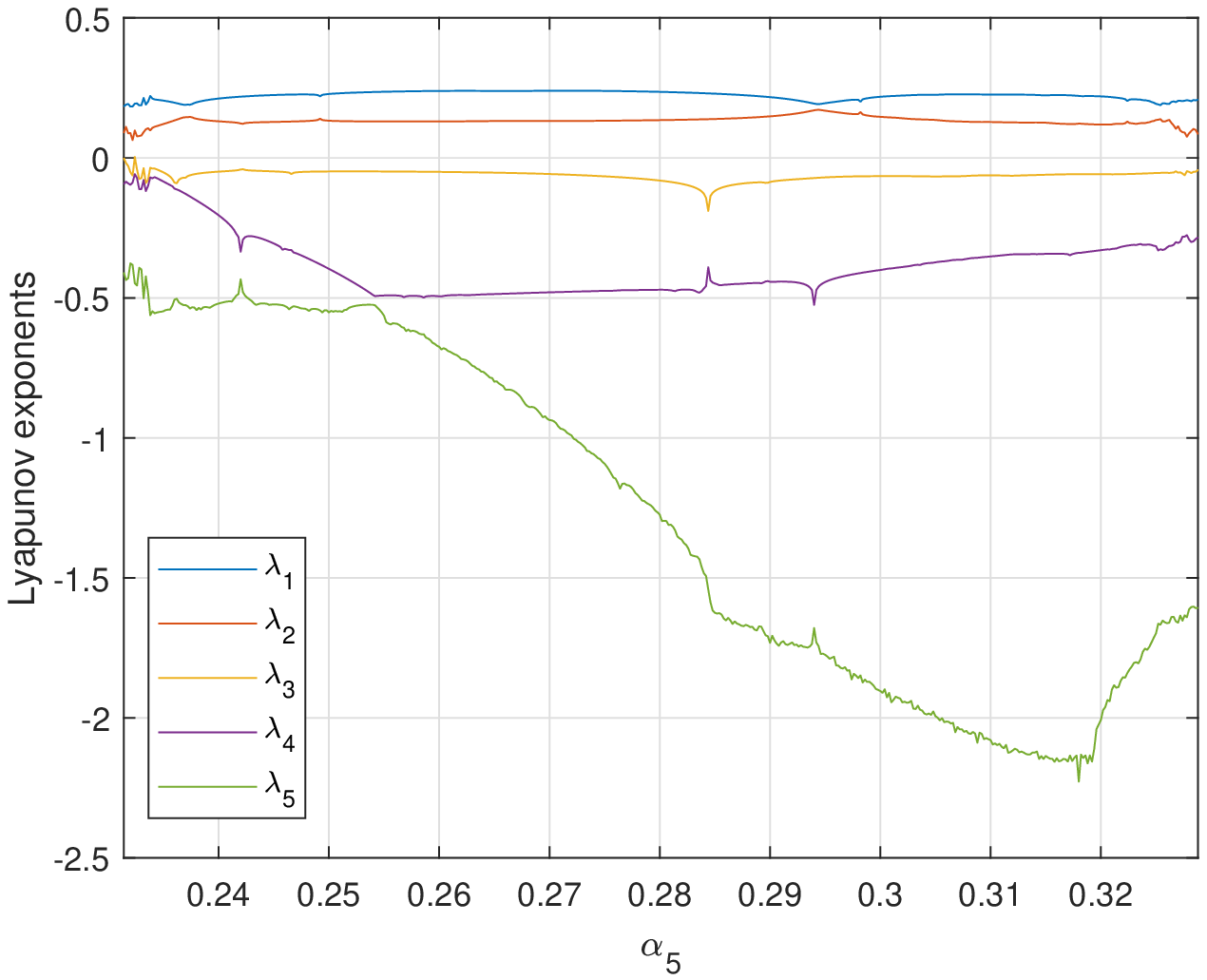}}}%
\subfigure[Bifurcation of u vs. $\alpha_5$.]{
\resizebox{7cm}{!}{\includegraphics{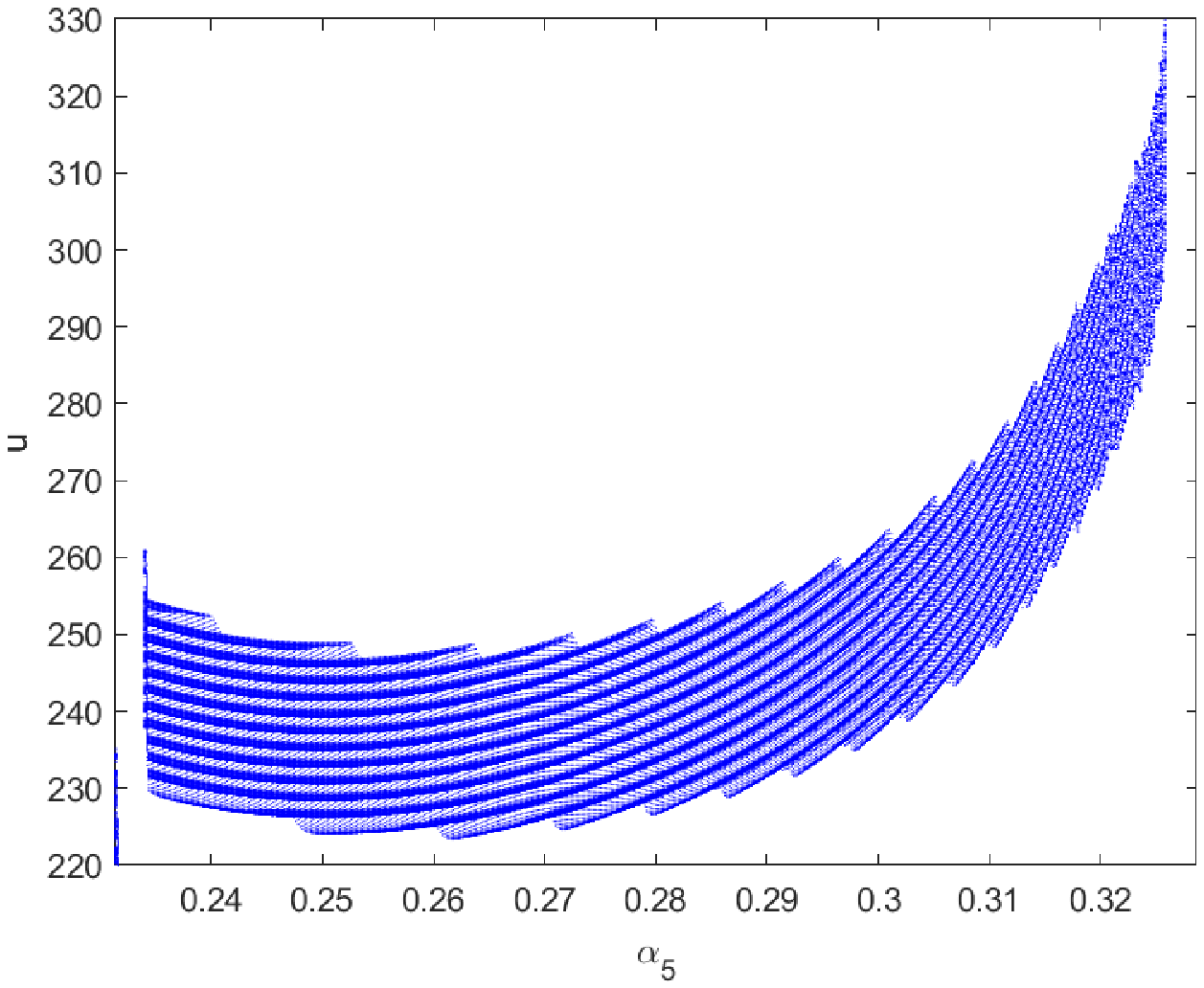}}}
\end{minipage}
\begin{minipage}{140mm}
\subfigure[Bifurcation of x vs. $\alpha_5$.]{
\resizebox{7cm}{!}{\includegraphics{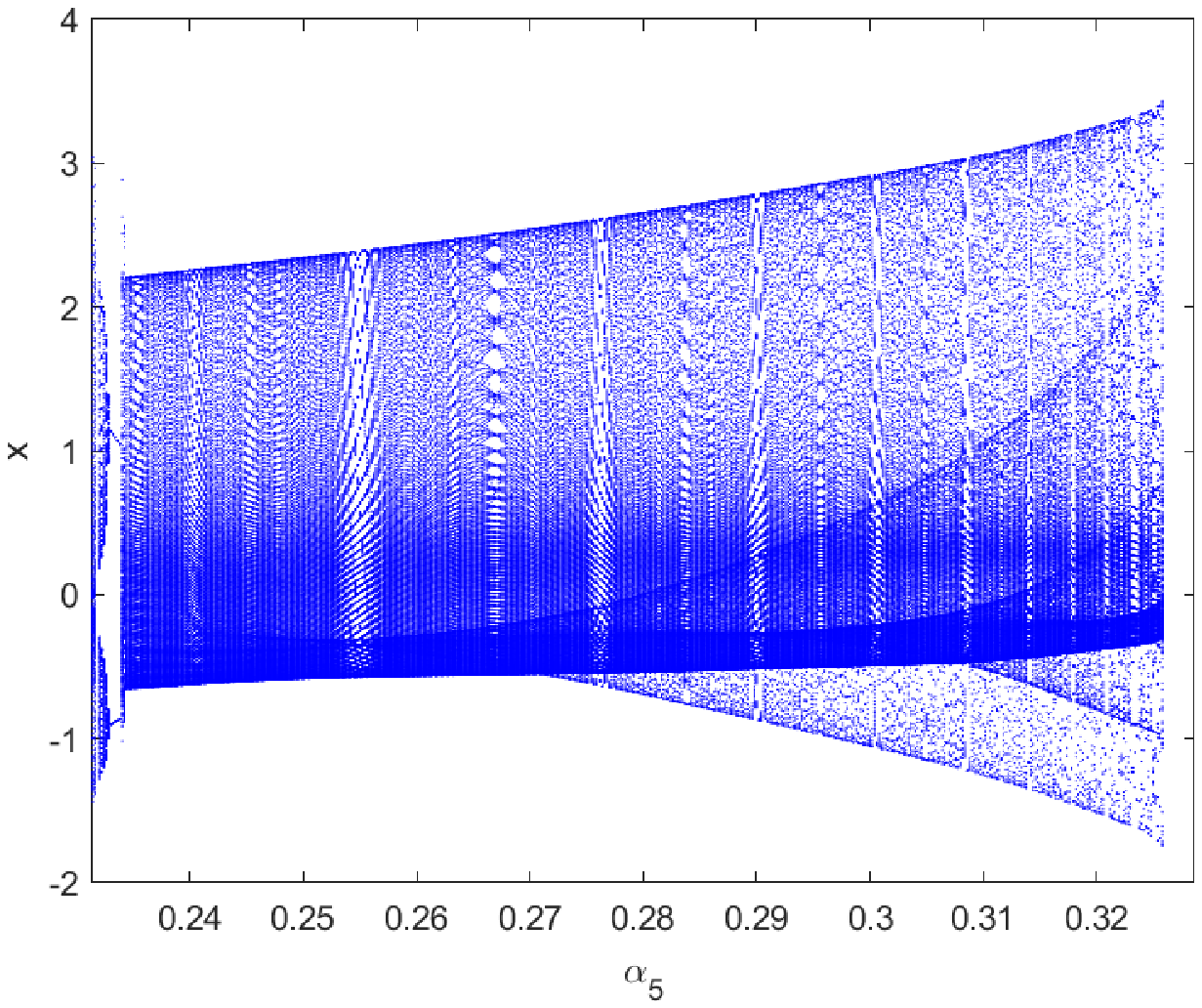}}}%
\subfigure[Bifurcation of z vs. $\alpha_5$.]{
\resizebox{7cm}{!}{\includegraphics{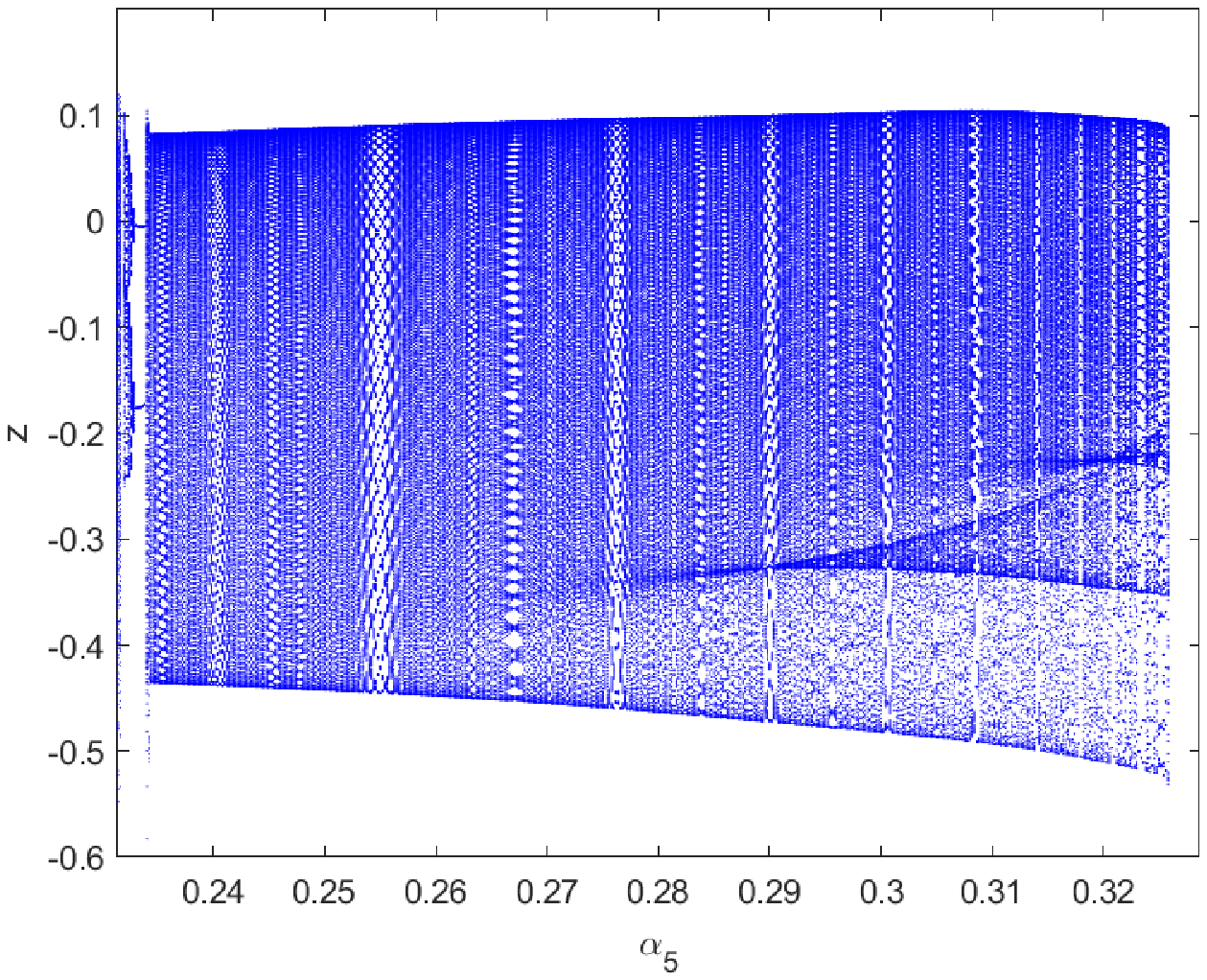}}}%
\end{minipage}
\begin{minipage}{140mm}
\subfigure[Hyperchaotic attractor in the (y, z, u)-plane vs. $\alpha_5=0.24$.]{
\resizebox{7cm}{!}{\includegraphics{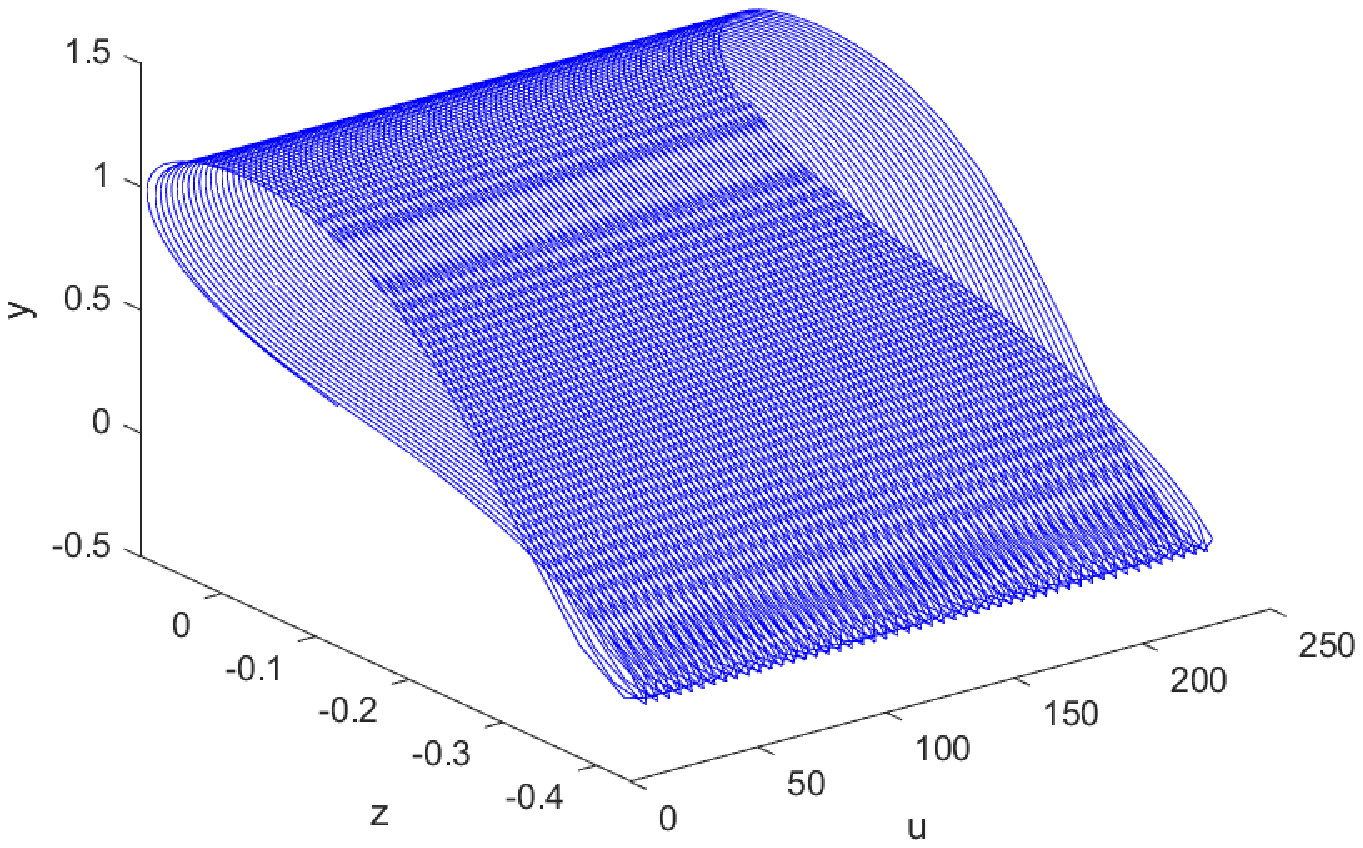}}}%
\subfigure[Hyperchaotic attractor in the (x, y, w)-plane vs. $\alpha_5=0.24$.]{
\resizebox{7cm}{!}{\includegraphics{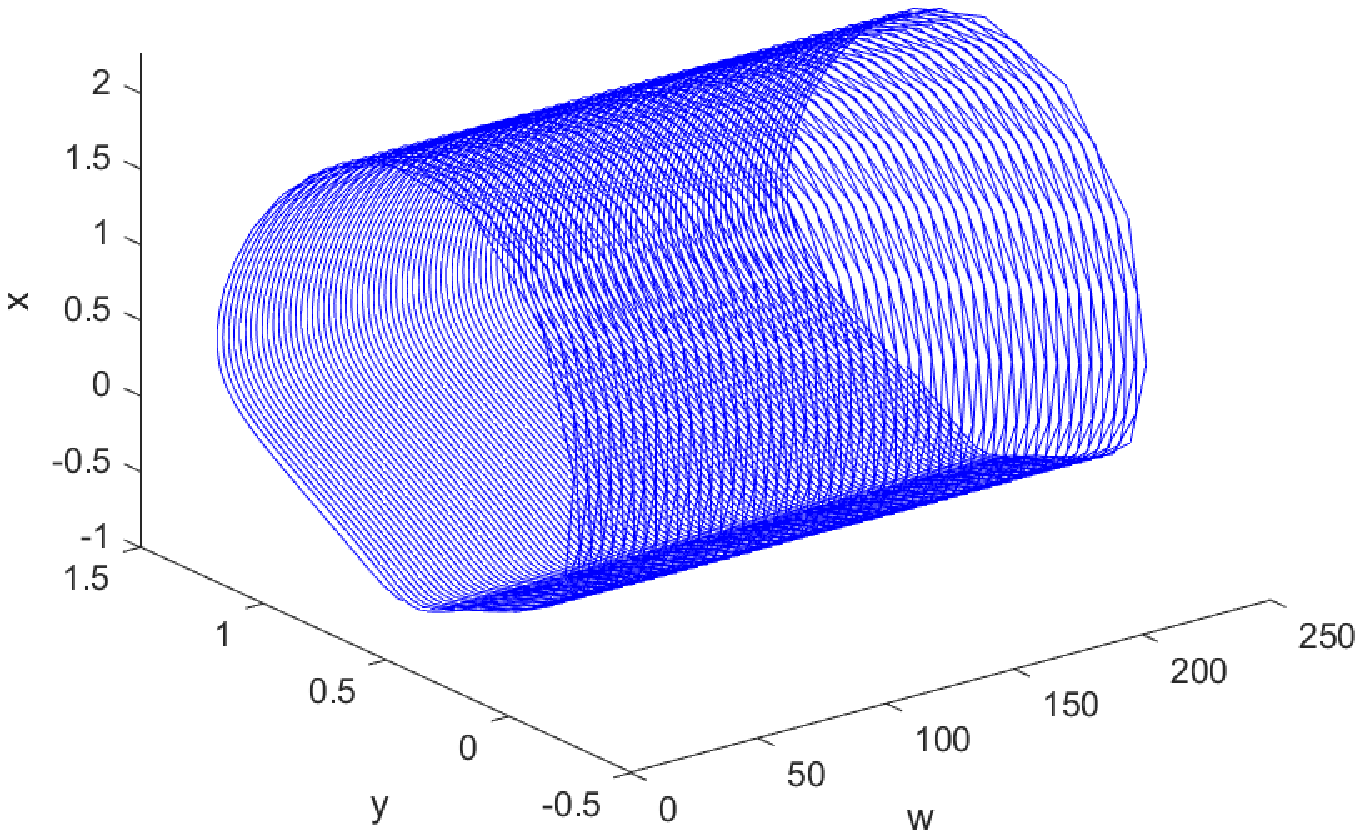}}}%
\caption{Complex behavior of system (\ref{eq4444}) with $\alpha_5$.}
\label{fg04}
\end{minipage}
\end{center}
\end{figure}\label{fg04}

\subsection{Varying $\alpha_5$ with fixed $k=2$ and $p=1$}
Using the iterative algorithm, we obtain a bifurcation diagram of the system (\ref{eq4444}) when we vary $\alpha_5$ and fix $k=2$ and $p=1$, as shown in Figure \ref{fg04}(a). In this Figure \ref{fg04}(a), we can always find two positive Lyapnov exponents corresponding to any $\alpha_5$, which is well validated by Figures \ref{fg04}(b-d). Thus, we can say that system (\ref{eq4444}) is hyperchaotic with $\alpha_5\in[0.232,\ 0.328]$. We can fix $\alpha_5=0.24$ and pick up a set of Lyapunov exponents $(\lambda_1,\ \lambda_2,\ \lambda_3,\ \lambda_4,\ \lambda_5)$=$(0.2120,\ 0.1308,\ -0.0491,\ -0.2048$, $-0.5201)$. Obviously there are two positive Lyapunov exponents $\lambda_1,\ \lambda_2$ and three negative Lyapunov exponents $\lambda_3,\ \lambda_4,\ \lambda_5$ when $\alpha_5=0.24$, i.e., there is a hyperchaotic attractor, as shown in Figures \ref{fg04}(e)-(f).

\subsection{Varying $p$ with fixed $k=2$ and $\alpha_5=0.3$}
Here, we employ iterative algorithm to produce a bifurcation diagram of system (\ref{eq4444}) when we vary $p$ and fix $k=2$ and $\alpha_5 =0.3$, as shown in Figure \ref{fg05}(a). In the Figure, we can always find two positive Lyapnov exponents corresponding to any $p$, which is well confirmed by Figures \ref{fg05}(b)-(d). Thus, we can say that system (\ref{eq4444}) is hyperchaotic with $p\in[1,\ 2]$. We can set $p=1$ and obtain a set of Lyapunov exponents $(\lambda_1,\ \lambda_2,\ \lambda_3,\ \lambda_4,\ \lambda_5)$=$(0.2182,\ 0.1468,\ -0.0653,\ -0.4007$, $-1.9060)$. There are aslo two positive Lyapunov exponents $\lambda_1,\ \lambda_2$ and three negative Lyapunov exponents $\lambda_3,\ \lambda_4,\ \lambda_5$ when $p=1$, i.e. there exists hyperchaos in system (\ref{eq4444}), as shown in Figures \ref{fg05}(e)-(f).
\begin{figure}
\begin{center}
\begin{minipage}{140mm}
\subfigure[Lyapunov exponents vs. $p$.]{
\resizebox{7cm}{!}{\includegraphics{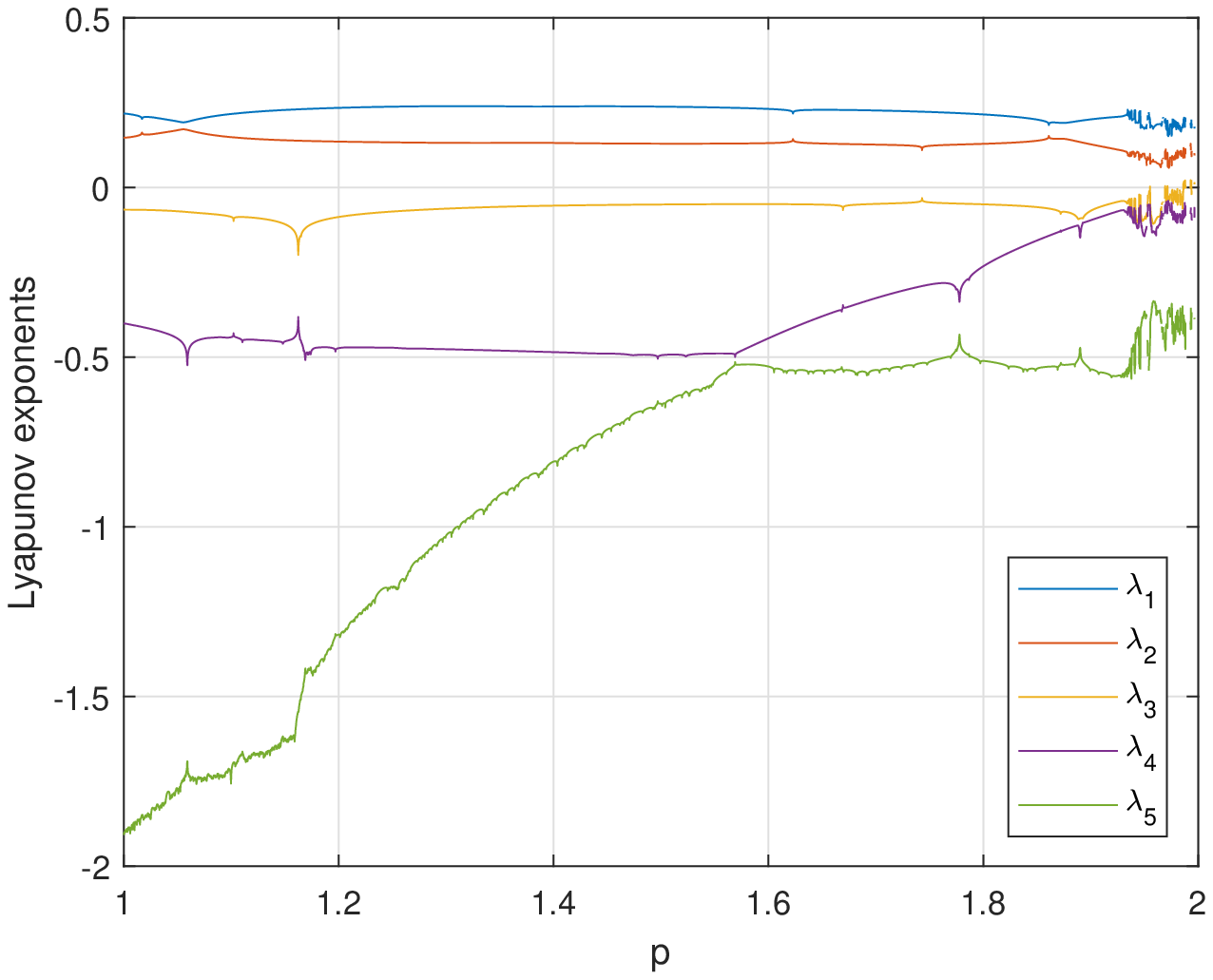}}}%
\subfigure[Bifurcation of u vs. $p$.]{
\resizebox{7cm}{!}{\includegraphics{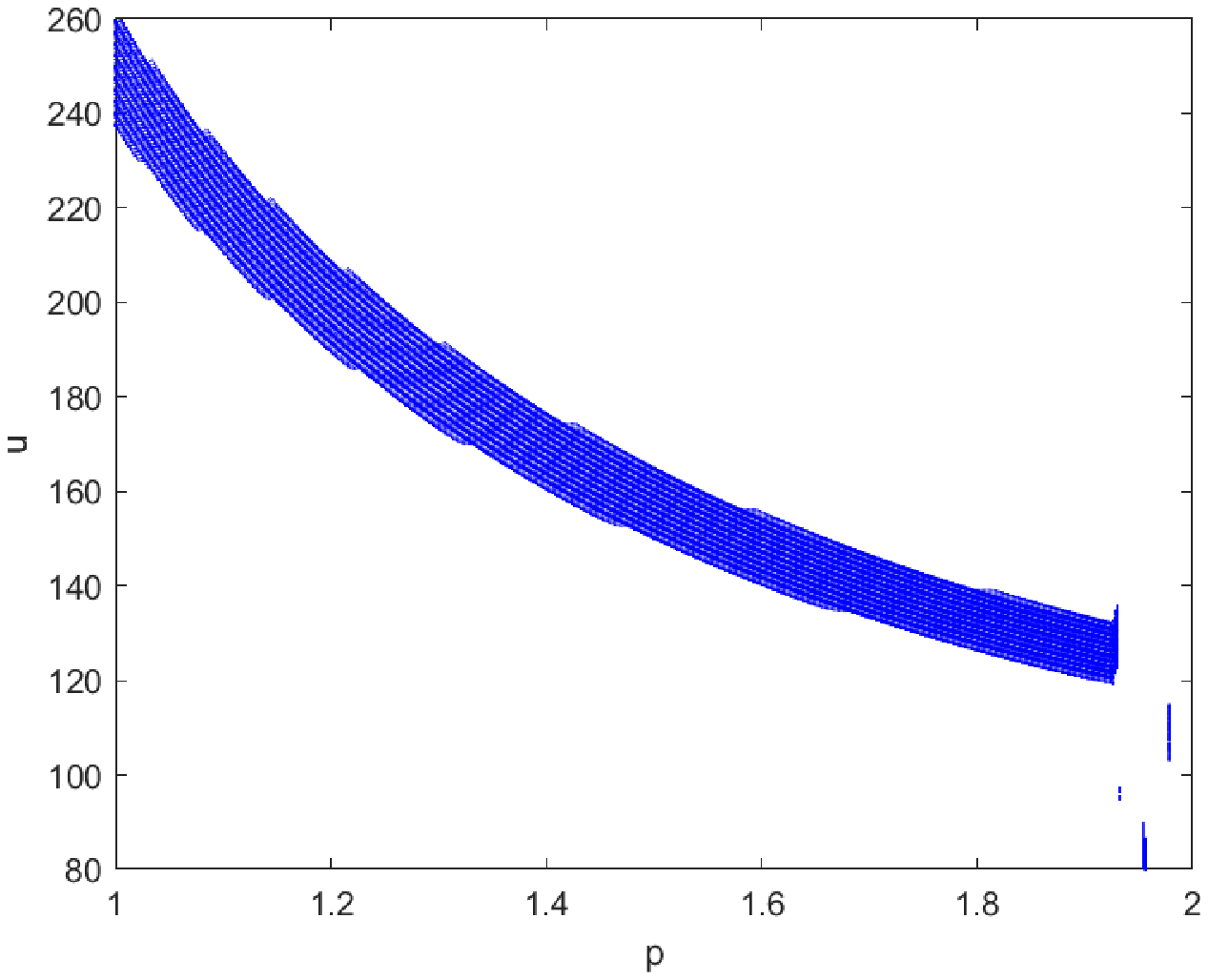}}}
\end{minipage}
\begin{minipage}{140mm}
\subfigure[Bifurcation of x vs. $p$.]{
\resizebox{7cm}{!}{\includegraphics{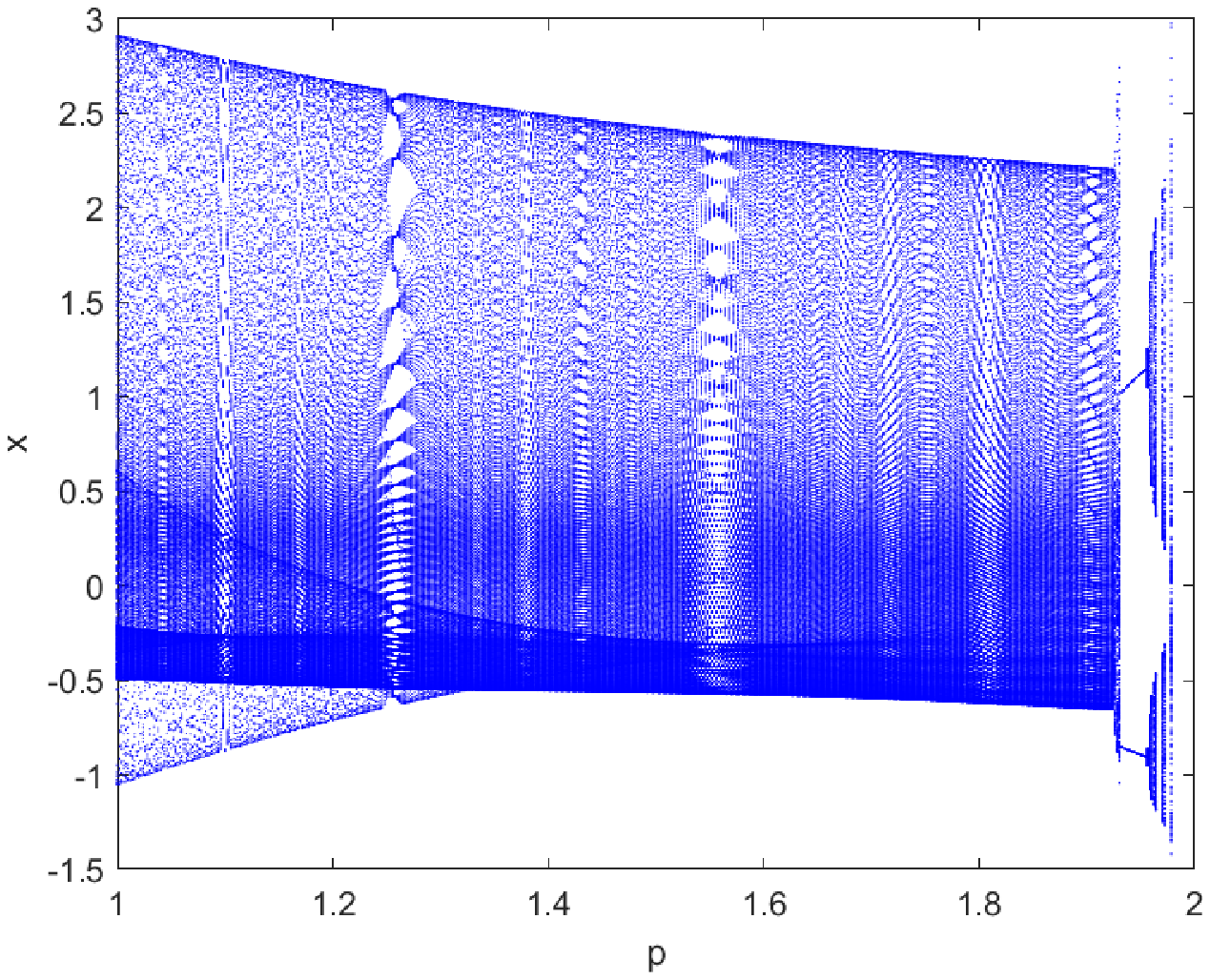}}}%
\subfigure[Bifurcation of z vs. $p$.]{
\resizebox{7cm}{!}{\includegraphics{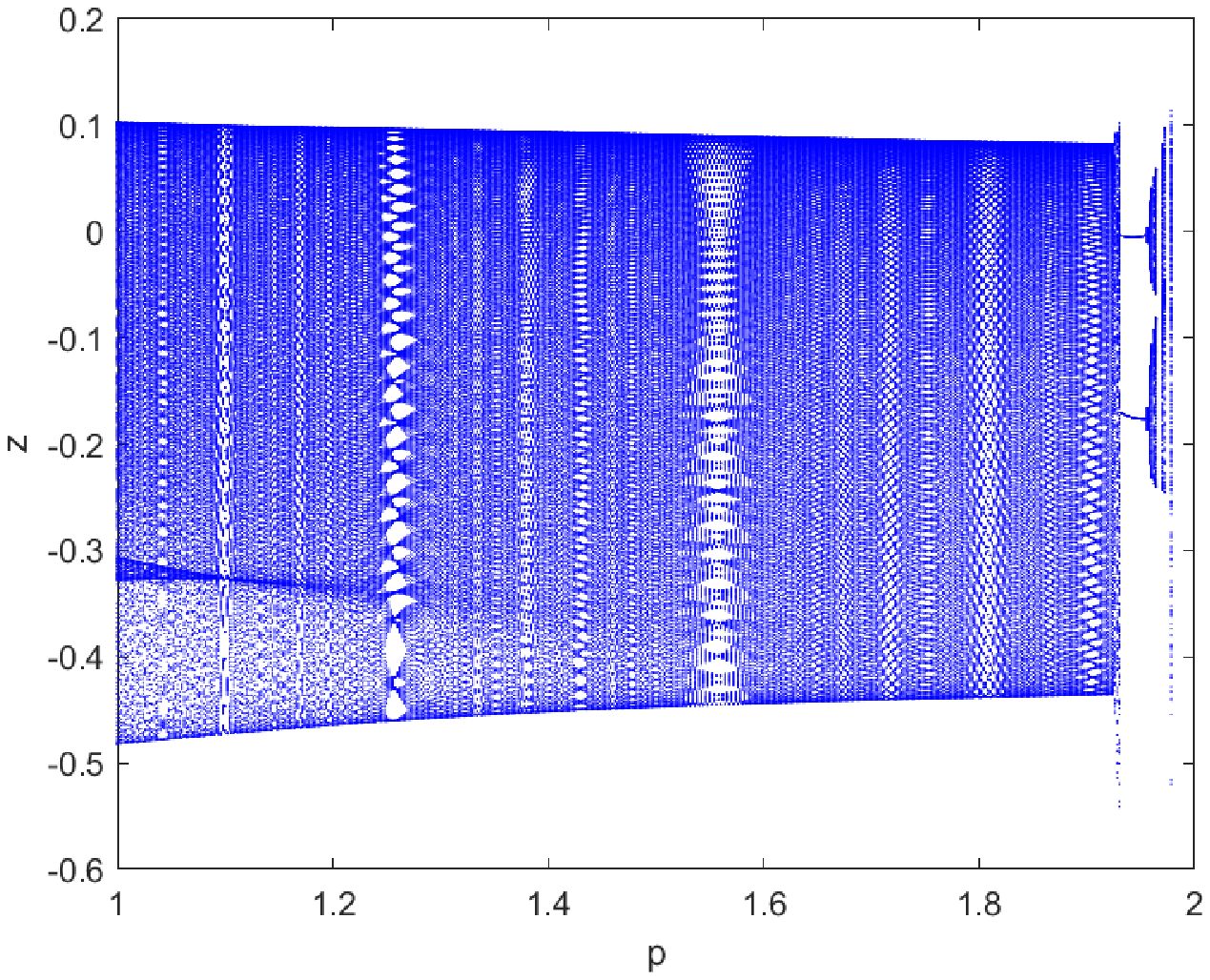}}}%
\end{minipage}
\begin{minipage}{140mm}
\subfigure[Hyperchaotic attractor in the (y, z, u)-plane vs. $p=1$.]{
\resizebox{7cm}{!}{\includegraphics{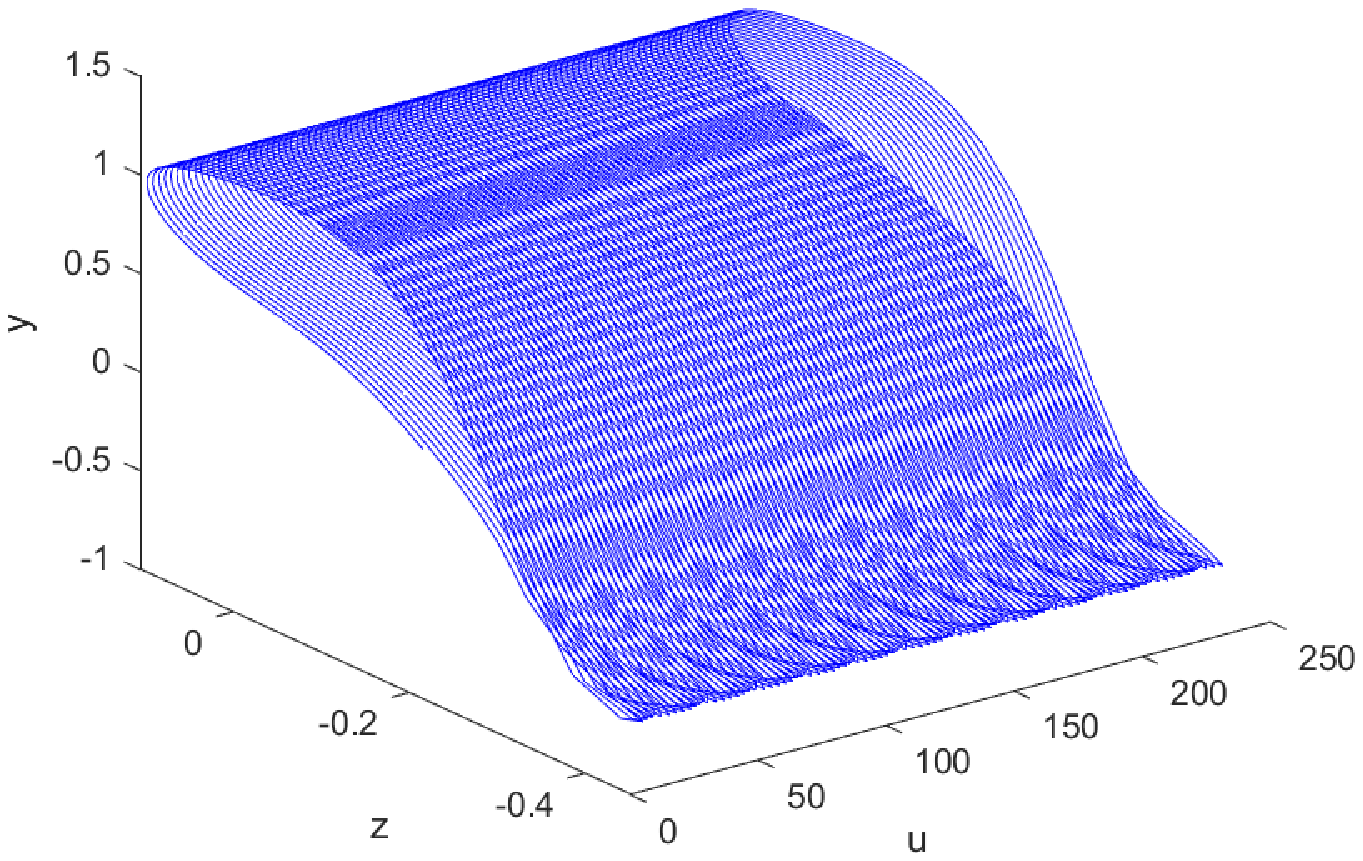}}}%
\subfigure[Hyperchaotic attractorin the (x, y, w)-plane vs. $p=1$.]{
\resizebox{7cm}{!}{\includegraphics{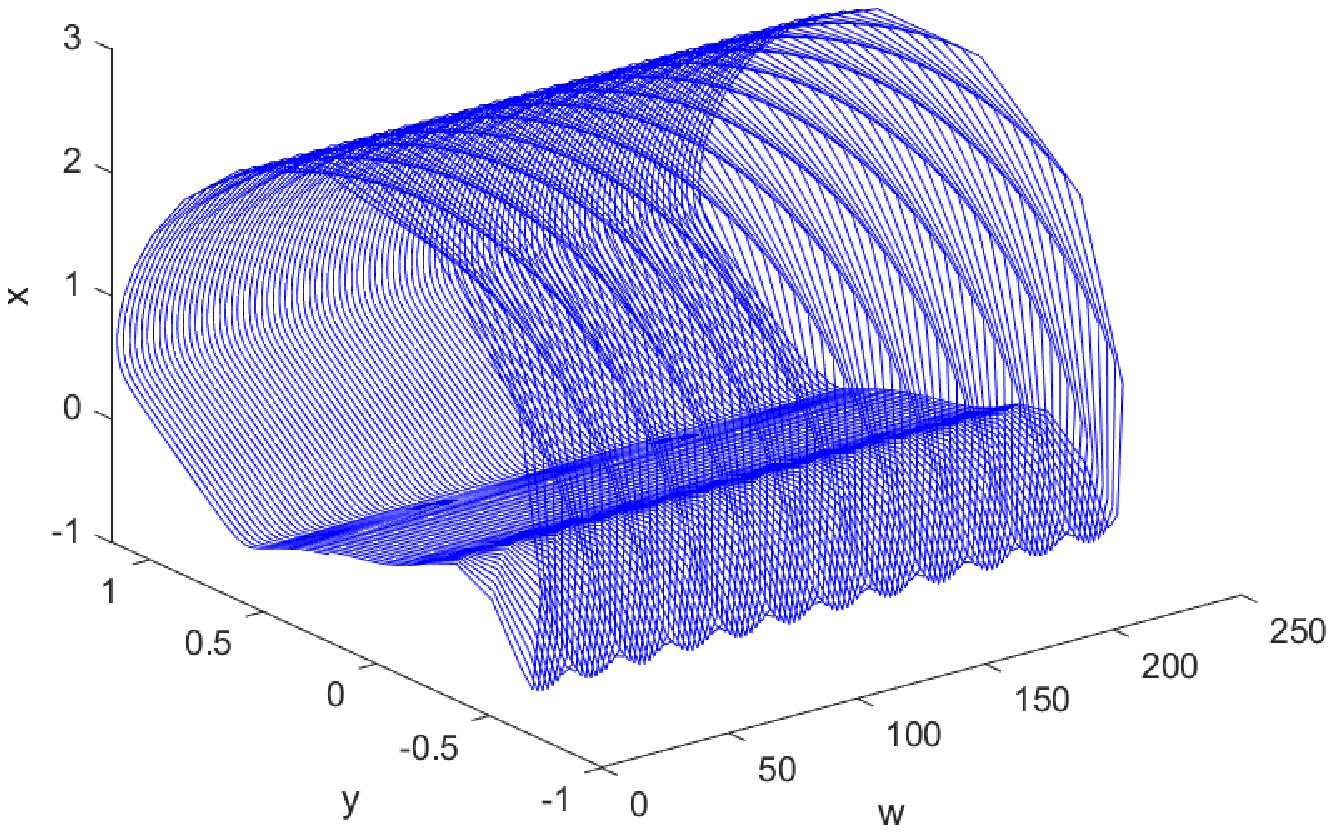}}}%
\caption{Complex behavior of system (\ref{eq4444}) with $p$.}
\label{fg05}
\end{minipage}
\end{center}
\end{figure}
\subsection{Varying $k$ with fixed $p=1$ and $\alpha_5=0.3$}
By means of the iterative algorithm, we draw a bifurcation diagram of system (\ref{eq4444}) when we vary $k$ and fix $p=1$ and $\alpha_5 =0.3$, as shown in Figure \ref{fg06}(a). In the Figure \ref{fg06}(a), we can always obtain two positive Lyapnov exponents corresponding to any $k$, which is well verified by Figures \ref{fg06}(b)-(d). Thus, we can say that system (\ref{eq4444}) is hyperchaotic with $k\in[1.5,\ 2.5]$. We also can let $k=1.5$ and obtain a set of Lyapunov exponents $(\lambda_1,\ \lambda_2,\ \lambda_3,\ \lambda_4,\ \lambda_5)$=$(0.1918,\ 0.1049,\ -0.0902,\ -0.3301$, $-1.5310)$. There are aslo two positive Lyapunov exponents $\lambda_1,\ \lambda_2$ and three negative Lyapunov exponents $\lambda_3,\ \lambda_4,\ \lambda_5$ when $k=1.5$, i.e. a hyperchaotic attractor occurs in system (\ref{eq4444}), as shown in Figures \ref{fg06}(e)-(f).

\begin{figure}
\begin{center}
\begin{minipage}{140mm}
\subfigure[Lyapunov exponents vs. $k$.]{
\resizebox{7cm}{!}{\includegraphics{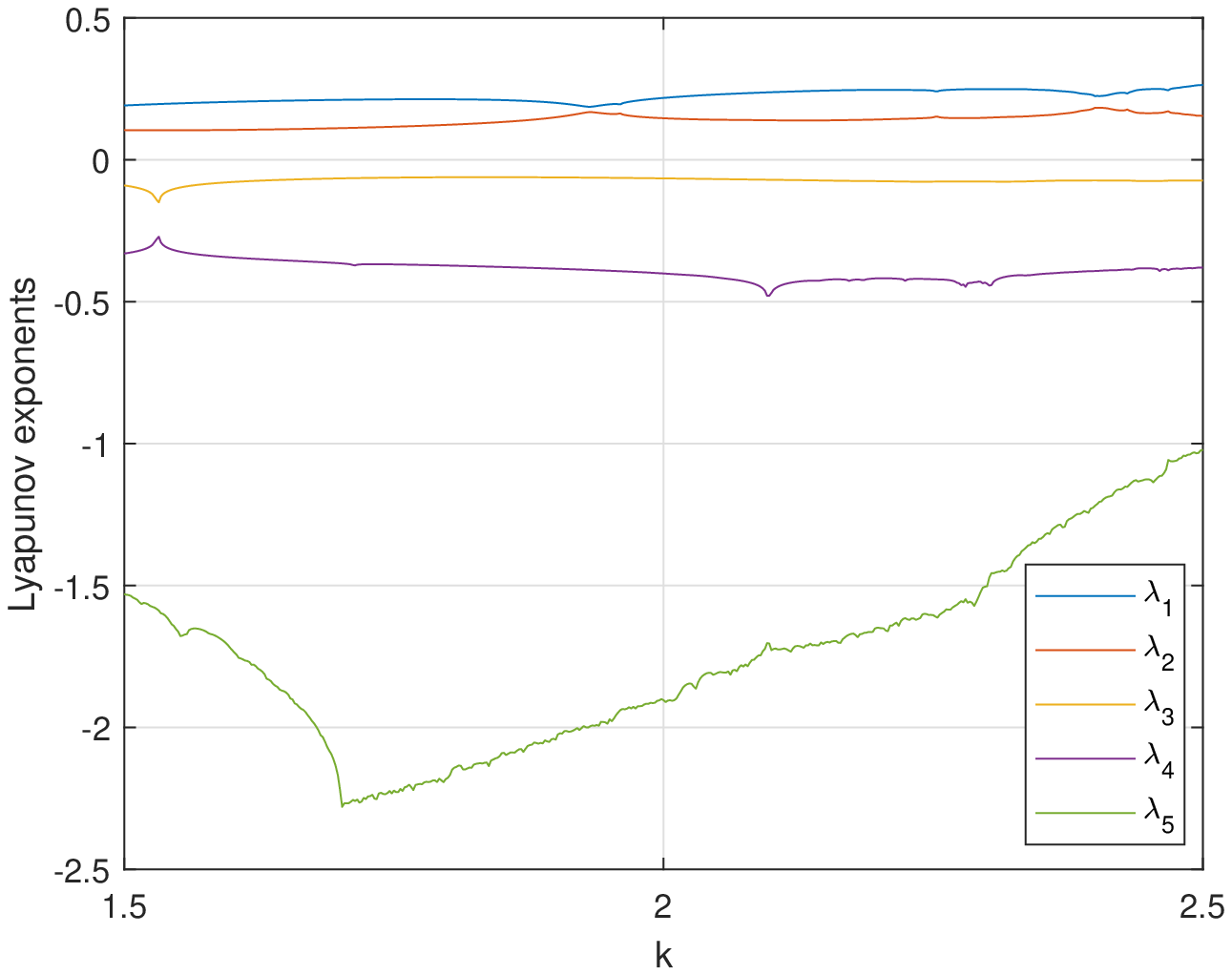}}}%
\subfigure[Bifurcation of u vs. $k$.]{
\resizebox{7cm}{!}{\includegraphics{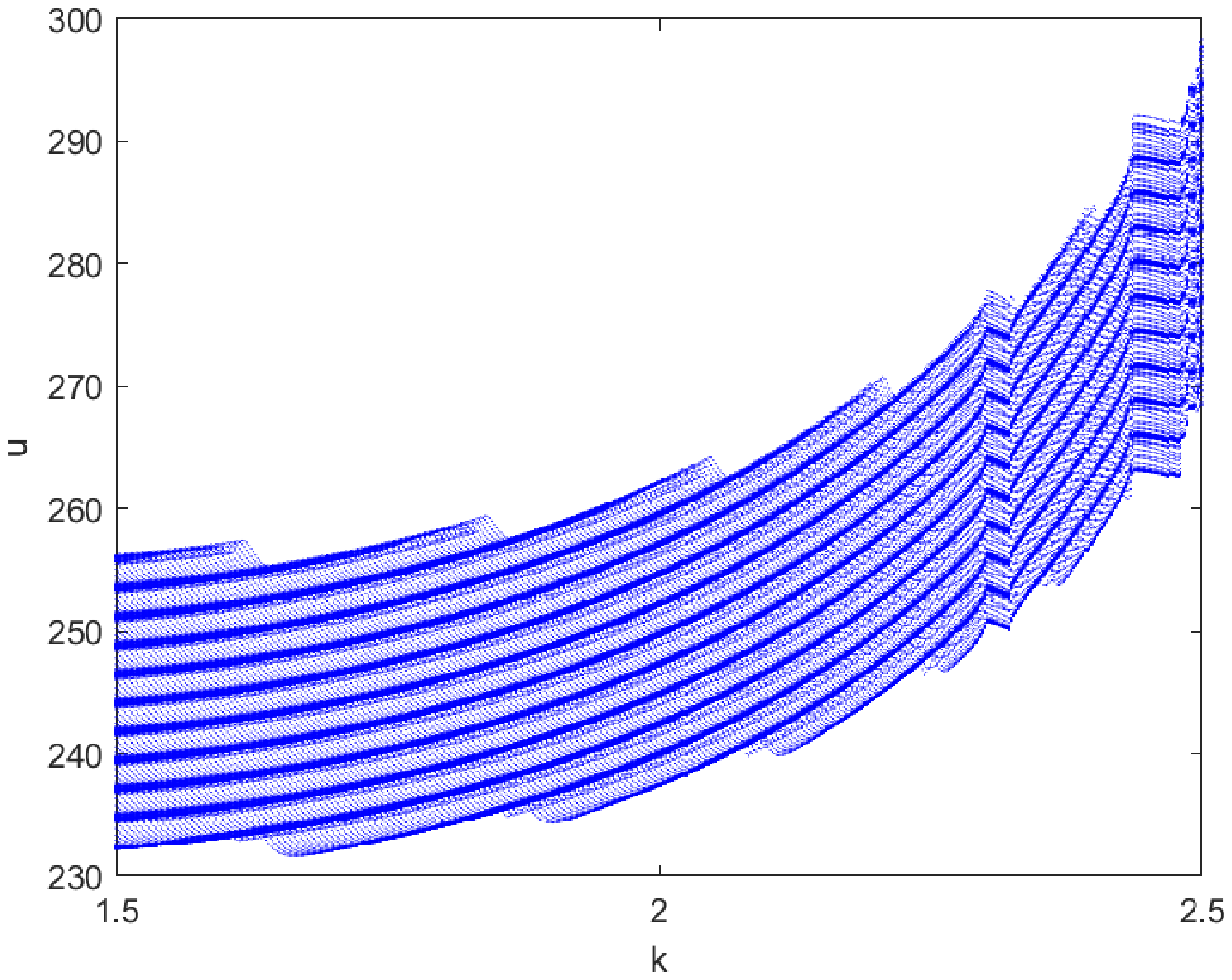}}}
\end{minipage}
\begin{minipage}{140mm}
\subfigure[Bifurcation of x vs. $k$.]{
\resizebox{7cm}{!}{\includegraphics{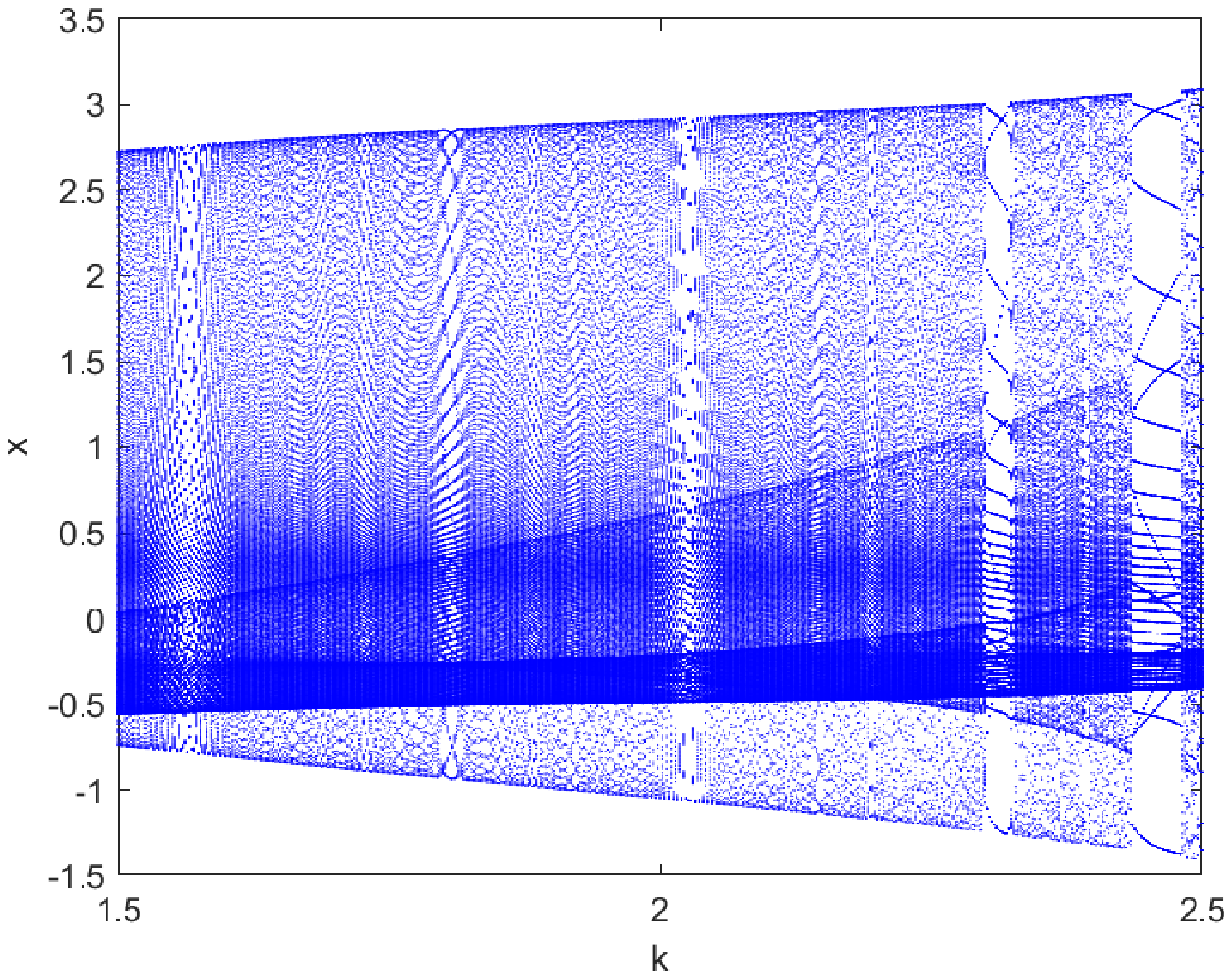}}}%
\subfigure[Bifurcation of z vs. $k$.]{
\resizebox{7cm}{!}{\includegraphics{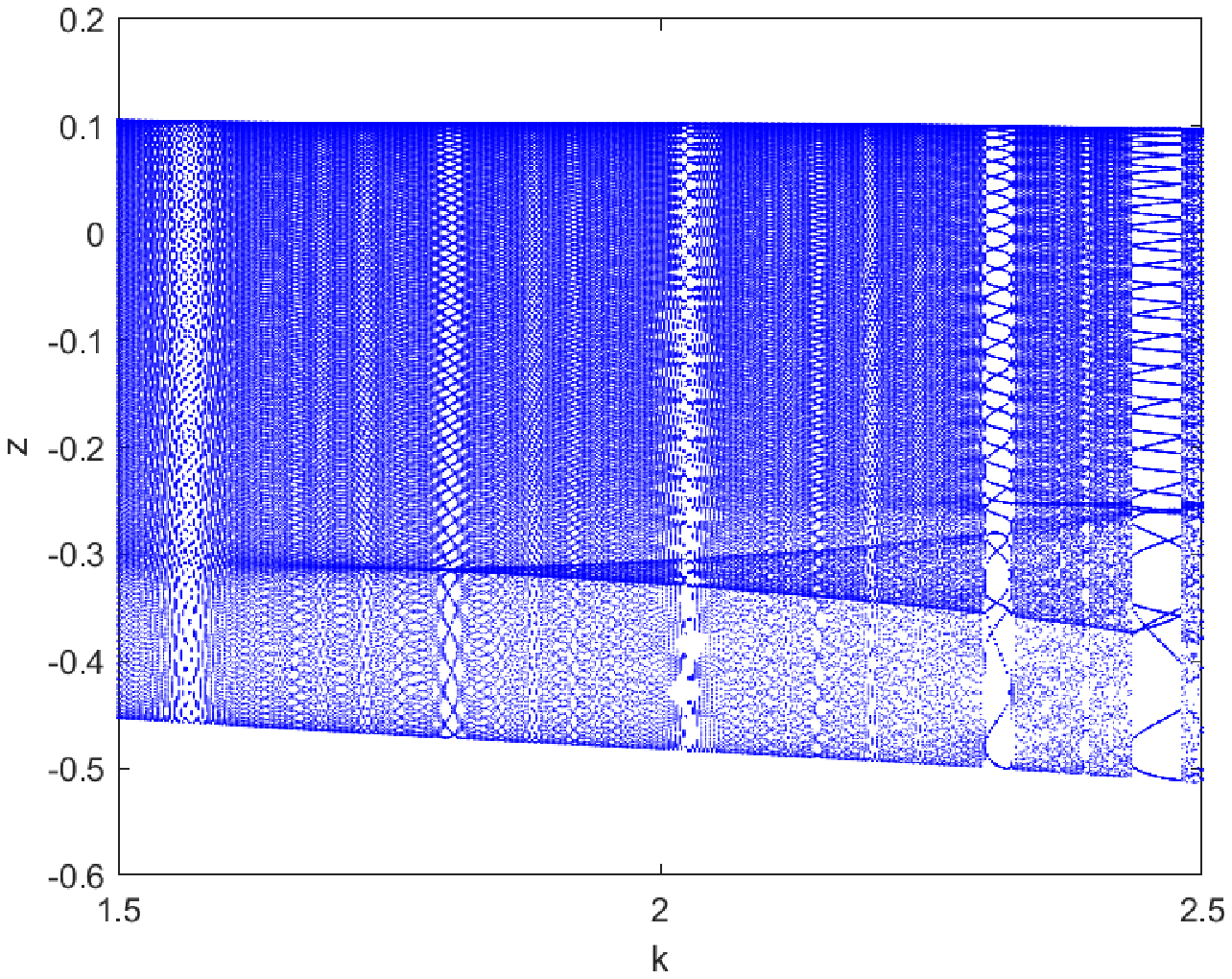}}}%
\end{minipage}
\begin{minipage}{140mm}
\subfigure[Hyperchaotic attractor in the (y, z, u)-plane vs. $k=1.5$.]{
\resizebox{7cm}{!}{\includegraphics{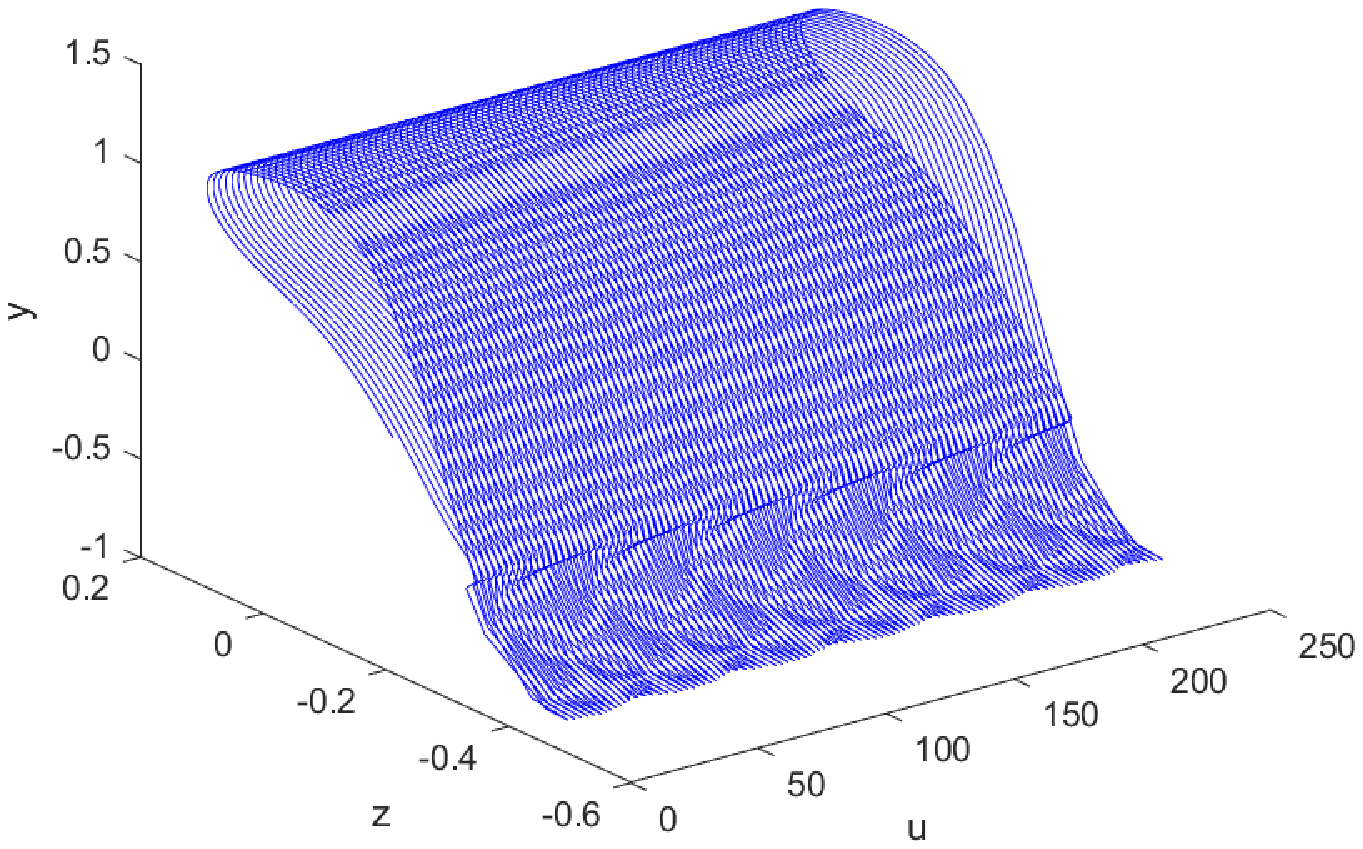}}}%
\subfigure[Hyperchaotic attractor in the (x, y, w)-plane vs. $k=1.5$.]{
\resizebox{7cm}{!}{\includegraphics{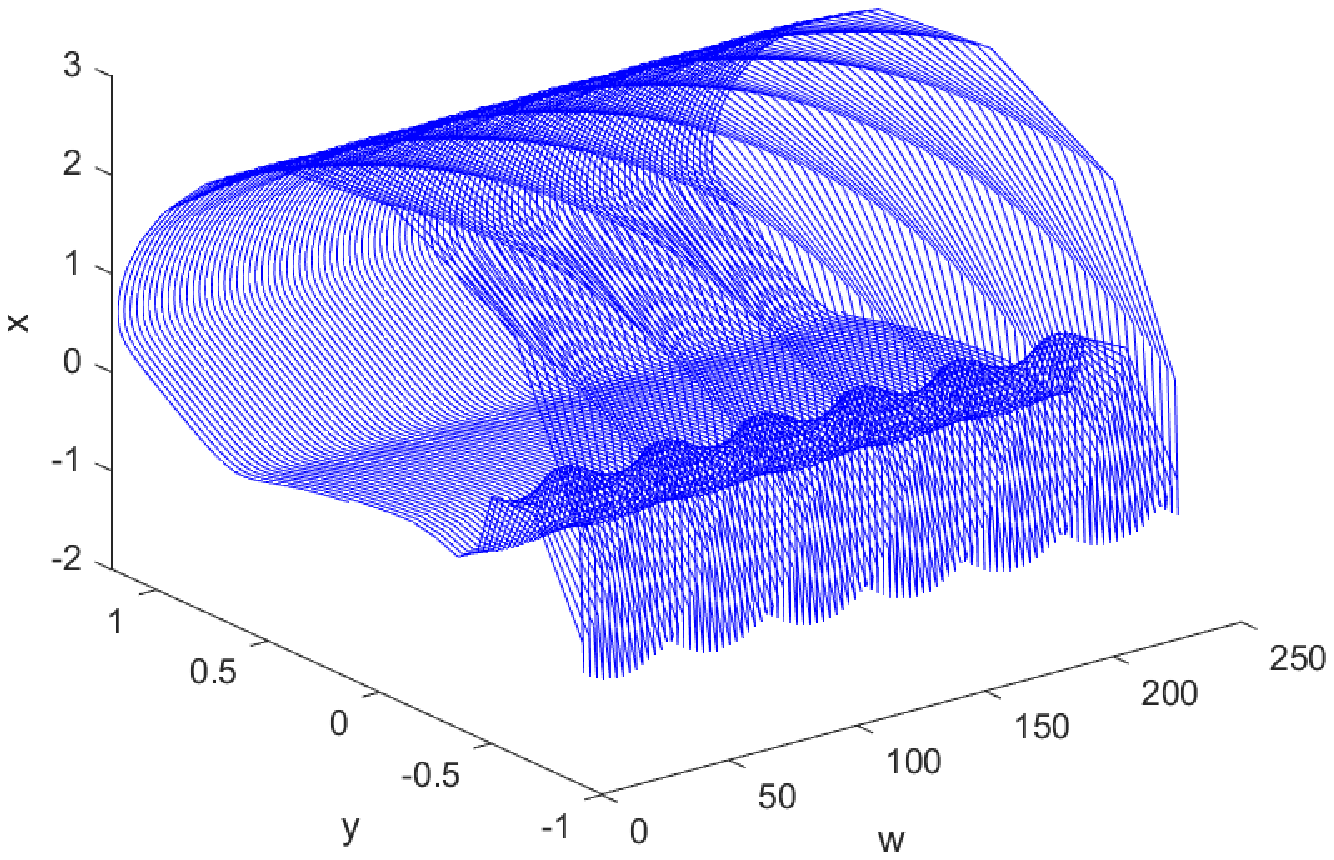}}}%
\caption{Complex behavior of system (\ref{eq4444}) with $k$.}
\label{fg06}
\end{minipage}
\end{center}
\end{figure}

\section{Conclusion}
As we know, a real financial system is very complicated, and it is impossible to accurately depict it. We try to introduce the ethics risk and conformable derivatives to an existing financial system with market confidence to set up a five-dimensional conformable derivative hyperchaotic financial system. To obtain the system's numerical solutions, we employ a piecewise constant approximation to design a new discretization process whose result is the same to that of conformable Euler's method. Finally, we use the proposed method to detect the hyperchaos in the proposed financial system.

The following extensions are of interest for future research.

(\rmnum{1})  Compared with traditional fractional-order derivatives, such as Riemann-Liouville and Caputo, the conformable derivative retains some most critical properties. A useful one is a relation between the conformable derivative and the classical integer ordered derivative in the definition of the composite function. So, we can regard the conformable derivative financial system as a natural extension of its ODE form. But the conformable derivative is a local derivative and does not have some remarkable properties, such as memory and nonlocality. Further development is to consider the memory for which a traditional fractional-order financial system with market confidence and ethics risk will be interesting.

(\rmnum{2}) In the real world, our government makes some decisions based on particular parameters of a real system. So it is essential for us to study the parameter estimation\cite{Yuanliguo2012,Behinfaraz2016,Belkhatir2018} of the conformable derivative financial system with market confidence and ethics.

(\rmnum{3}) Though we draw some dynamical system model and implement some evolutionary analysis by conformable derivative, it will be more interesting to verify the proposed model by experimental economics approach\cite{Pikulina2017,Deaves2018}.

(\rmnum{4}) The proposed discretization process for conformable derivative systems is easy to be implemented and can be used in many scientific fields, such as engineering, physics, economics, environment, ecology, and materials science.

\section*{Acknowledgments}
\small The authors would like to express many thanks to referees for their time and efforts in reviewing this manuscript. Their helpful comments and constructive suggestions greatly improved this manuscript.

\section*{Funding}
\small This work is supported partly by Natural Science Foundation of Shandong Province (Grant No. ZR2016FM26),  National Social Science Foundation of China (Grant No.16FJY008), and National Natural Science Foundation of China (Grant No.11801060).

\section*{Availability of data and materials}
\small Not applicable.

\section*{Competing interests}
\small The authors declare that they have no competing interests.

\section*{Authors’ contributions}
\small The authors have made the same contribution. All authors read and approved the final manuscript.

\end{document}